\documentclass[11pt,a4paper]{article}
\usepackage[english]{babel}
\usepackage[utf8]{inputenc}
\usepackage[T1]{fontenc}
\usepackage{mathrsfs}
\usepackage[leqno]{amsmath}
\usepackage{color}
\usepackage{comment}
\usepackage{stmaryrd} %for \llbracket

\usepackage[a4paper,top=1.5in,bottom=1.5in,left=1in,right=1in,marginparwidth=1.75cm]{geometry}

\usepackage{amsmath,amssymb}
\usepackage{bm} %nice bold maths
\usepackage{enumerate} %allows us to use roman numerals for lists
\usepackage{amsthm}
\usepackage{mathtools} %gives :=

\usepackage{filemod} %autoupdate last edited date

\usepackage{pdflscape}
\usepackage{afterpage}
\makeatletter
\newcommand{\leqnos}{\tagsleft@true\let\veqno\@@leqno}
\newcommand{\reqnos}{\tagsleft@false\let\veqno\@@eqno}
\reqnos
\makeatother

\mathtoolsset{showonlyrefs,showmanualtags}%only numbers labelled equations

\theoremstyle{plain}
\newtheorem{thm}{Theorem}[section] % reset theorem numbering for each chapter
\newtheorem{proposition}[thm]{Proposition} % same for Propositions
\newtheorem{cor}[thm]{Corollary} % same for Corollaries
\theoremstyle{definition}
\newtheorem{defn}[thm]{Definition} % definition numbers are dependent on theorem numbers

\theoremstyle{remark}
\newtheorem{remark}[thm]{Remark}
\newtheorem{example}[thm]{Example} % same for example numbers

\newcommand{\R}{\mathbb{R}}

\renewcommand{\theta}{\vartheta}
\renewcommand{\epsilon}{\varepsilon}
\renewcommand{\P}{\mathbb{P}}

\newcommand{\E}{\mathbb{E}}

\newcommand{\NN}{\mathbb{N}}

\newcommand{\RR}{\mathbb{R}}

\newcommand{\cA}{\mathscr{A}}

\newcommand{\cC}{\mathscr{C}}
\newcommand{\cE}{\mathcal{E}}
\newcommand{\cF}{\mathcal{F}}

\newcommand{\diff}{\mathrm{d}}
\newcommand{\dd}{\,\mathrm{d}}

\newcommand{\1}{\mathbf{1}}

\newcommand*{\EX}[2][]{\E^{#1}\left [ #2 \right ]}
\newcommand*{\cEX}[3][]{\E^{#1}\left[ #2 \,\middle\vert\, #3 \right]}

\newcommand*{\as}[1]{#1\text{-a.s.}}

\title{An elementary approach to the Merton problem}
\author{Martin Herdegen\thanks{All authors: University of Warwick, Coventry CV4 7AL, UK; m.herdegen, d.hobson, j.jerome@warwick.ac.uk},  David Hobson\thanks{Communicating author}, Joseph Jerome\thanks{We thank Steve Shreve and Ioannis Karatzas for sharing their recollections of their motivations behind~\cite{lehoczky1983optimal} and \cite{karatzas1986explicit}. \newline Data sharing not applicable – no new data generated.}}

\date{\today}

\begin{document}
\maketitle

\begin{abstract}
    In this article we consider the infinite-horizon Merton investment-consumption problem in a constant-parameter Black\textendash Scholes\textendash Merton market for an agent with constant relative risk aversion $R$. The classical primal approach is to write down a candidate value function and to use a verification argument to prove that this is the solution to the problem. However, features of the problem take it outside the standard settings of stochastic control, and the existing primal verification proofs rely on parameter restrictions (especially, but not only, $R<1$), restrictions on the space of admissible strategies, or intricate approximation arguments.

    The purpose of this paper is to show that these complications can be overcome using a simple and elegant argument involving a stochastic perturbation of the utility function.
\end{abstract}

\section{Introduction and overview}
In the Merton investment-consumption problem (Merton~\cite{merton1969lifetime,merton1971optimum}) an agent seeks to maximise the expected integrated discounted utility of consumption over the infinite horizon in a model with a risky asset and a riskless bond. When parameters are constant and the utility function is of power type, it is straightforward to write down the candidate value function. However, it is more difficult to give a complete verification argument. For general strategies the wealth process may hit zero at which point the application of It\^{o}'s formula to the candidate value function breaks down; the local martingale which arises from the application of It\^{o}'s formula may fail to be a martingale; even for constant proportional strategies transversality may fail.

For all these reasons, it is difficult to give a concise, rigorous verification proof via analysis of the value function, and many textbooks either finesse the issues or restrict attention to a subclass of admissible strategies, and/or restrict attention to a subset of parameter combinations (especially $R<1$, but even then there can be substantive points which are often overlooked). The need for such a verification argument has been obviated by the development of proofs using the dual method, which provides a powerful and intuitive alternative approach, see Biagini~\cite{biagini2010expected} for a survey (and also \cite{karatzas1989optimization,karatzas1991martingale,karatzas1998methods,rogers2013optimal}). Nonetheless, it would be nice to provide a short proof based on the primal approach.\footnote{Our original motivation for returning to the Merton problem arose from consideration of a problem involving stochastic differential utility. There, dual approaches are more involved and do not cover all parameter combinations, so that the primal method is not redundant, and indeed may provide a more direct approach.} The goal of this paper is didactic -- to give a simple, brief proof that the candidate value function is the value function via the primal approach, and moreover, to give a proof which is valid for all parameter combinations for which the Merton problem is well-posed.

The first full verification of the solution to the Merton problem of which we are aware (under an assumption of positive discounting and strictly positive interest rates) is Karatzas et al~\cite{karatzas1986explicit}, which built on the previous work of Lehoczky et al~\cite{lehoczky1983optimal}. There, the idea is to solve a perturbation of the original problem in which the agent may go bankrupt, at which point they receive a residual value $P$. (Part of their motivation was to better understand the results of Merton~\cite{merton1971optimum} on HARA utilities, see also Sethi and Taskar~\cite{SethiTaskar:88}.) The solution to the perturbed problem is very clever, and is developed in the case of a general utility function, but it is also very intricate and takes many pages of calculation. Moreover, when specialised to the case of CRRA utilities, it places some assumptions on the parameter values beyond the necessary assumption of well-posedness of the Merton problem. The problem with bankruptcy is of independent interest, but more important for our purposes is the fact that, given the solution to the problem with bankruptcy for a CRRA utility, by letting $P \downarrow 0$ ($R<1$) or $P \downarrow -\infty$ ($R>1$) Karatzas et al~\cite{karatzas1986explicit} recover the solution to the original Merton problem.

In their seminal paper on transaction costs, Davis and Norman~\cite[Section 2]{davis1990portfolio} briefly consider the Merton problem without transaction costs. They assume that the proportion of wealth invested in the risky asset is bounded, and for $R<1$ they go on to prove a verification theorem for strategies restricted to this class. Further, in the case $R>1$ they propose a different perturbation, this time a deterministic perturbation of the candidate value function. The key point is that in the perturbed problem the candidate value function has a finite lower bound, and this allows Davis and Norman~\cite{davis1990portfolio} to re-apply arguments from the $R<1$ case, although the restriction to `regular' investment strategies remains. The candidate value for the perturbed problem can be used to give an upper bound on the true value function, which converges to the candidate solution to the Merton problem as the perturbation disappears. Unlike the argument in Karatzas et al~\cite{karatzas1986explicit}, the proof is quite short, but again it only works for certain parameter combinations, and more importantly it restricts attention to a subclass of admissible strategies.

Our goal is to give a complete, simple verification argument via primal methods. At its heart,
our idea is a modification of the approach in \cite{davis1990portfolio}. We perturb the utility function, which leads to a perturbed value function. However, rather than perturbing by the addition of a deterministic constant, we perturb by adding a multiple of the optimal wealth process. The great benefit is that the optimal consumption and the optimal investment are unchanged under the perturbation, which means that mathematical calculations remain strikingly simple. Moreover, these arguments are valid whenever the Merton problem is well-posed.

This paper is structured as follows. In the next section we introduce the problem, and in Section~\ref{sec:candidate} we give the candidate value function.
In Section~\ref{sec:fiat} we give a proof of the main result under a set of clearly-stated assumptions which are designed precisely to make the proof work.
Often, proofs in the stochastic control literature (see, for example, Davis and Norman~\cite{davis1990portfolio}, Fleming and Soner~\cite[Example 5.2]{fleming2006controlled} and Pham~\cite{pham:06}) artificially impose restrictions on the set of admissible strategies or on the parameter values to ensure that these assumptions are satisfied by default. In Section~\ref{sec:veri} we give our proof, which works for all parameter combinations and allows for all admissible strategies.
Finally, in a series of appendices we: first, give an example which illustrates how one of the clearly-stated assumptions may easily fail; second, give a small amount of detail on the Karatzas et al~\cite{karatzas1986explicit} and Davis and Norman~\cite{davis1990portfolio} approaches to the verification problem; third, discuss the case of logarithmic utility; fourth, consider the Merton problem under a change of numéraire and discuss the role of the parameter $\delta$; and fifth, for completeness, give a brief discussion of duality methods for the Merton problem.

%In Section~\ref{sec:literatureR>1} we give a small amount of detail on the Karatzas et al~\cite{karatzas1986explicit} and Davis and Norman~\cite{davis1990portfolio} approaches to the verification problem before in Section~\ref{sec:veri} we give our argument. Finally, in Section~\ref{sec:numeraire} we explain the insights which arise from considering a num\'{e}raire change of the problem, in particular on what are the appropriate parameter restrictions, and what is the best formulation of the problem. We also show how a change of num\'{e}raire can lead to a simplified solution of a slightly modified version of the problem with bankruptcy considered in Karatzas et al~\cite{karatzas1986explicit}, and hence to a simplified verification argument.

Our proof is an improvement on the existing primal results in at least three important ways.
First, it places no restrictions on the class of admissible strategies: for example, unlike much of the stochastic control literature, it does not require the fraction of wealth invested in the risky asset to be bounded. (The argument in Karatzas et al~\cite{karatzas1986explicit} also applies to general investment strategies.) Second, the proof covers all parameter combinations for which the Merton problem is well-posed (and does not assume that interest rates and discounting are positive -- as we shall argue these quantities depend on the choice of accounting units, and therefore are not absolutes in themselves). Third, our proof is simple, elegant and concise and not counting the derivation of the candidate solution and candidate value function can be written up in just over one page (Theorem~\ref{thm:veri} and Corollary~\ref{cor:veri}). %One final contribution of this article is to argue that some formulations of the Merton problem are more natural than others, in the sense that they are robust to changes of currency unit, and in consequence have simpler dependencies on parameter combinations.

\section{The Merton problem}
\label{sec:merton}

Throughout this paper we will work on a filtered probability space $(\Omega, \mathcal{F} , \mathbb{P}, \mathbb{F}=(\mathcal{F}_t)_{t > 0}) $ satisfying the usual conditions and supporting a Brownian motion $W = (W_t)_{t \geq 0}$. We will assume a Black\textendash Scholes\textendash Merton financial market consisting of a risk-free asset with interest rate $r \in \RR$  whose price process $S^0=(S^0_t)_{t \geq 0}$ is given by $S^0_t = \exp(r t)$ and a risky asset whose price process $S= (S_t)_{t \geq 0}$ follows a geometric Brownian motion with drift $\mu \in \RR$ and volatility $\sigma > 0$:
$$\frac{dS_t}{S_t} =  \mu \dd t + \sigma \dd W_t ,\quad S_0 = s > 0.$$

An agent operating with this investment opportunity set and initial wealth $x > 0$ chooses an \emph{admissible investment-consumption strategy} $(\theta^0, \theta, C) = (\theta^0_t, \theta_t, C_t)_{t \geq 0}$, where $\theta^0_t \in \RR$ denotes the number of riskless assets held at time $t$, $\theta_t \in \RR$ denotes the number of shares held at time $t$, and $C_t \in \R_+$ represents the rate of consumption at time $t$. We require that $\theta^0, \theta, C$ are progressively measurable processes, $\theta^0$ is integrable with respect to $S^0$, $\theta^1$ is integrable with respect to $S$, $C$ is integrable with respect to the identity process\footnote{By saying that a process $X$ is integrable with respect to the identity process we mean that $\int_0^t |X_s| ds < \infty$ $\as{\P}$ for each $t > 0$.}, the \emph{wealth process} $X = (X_t)_{t \geq 0}$ defined by
\begin{equation}
\label{eq:wealth process}
X_t := \theta^0_t S^0_t + \theta_t S_t
\end{equation}
is $\as{\P}$ nonnegative and the self-financing condition,
\begin{equation}
\label{eq:SF}
X_t = x + \int_0^t \theta^0_s \dd S^0_s + \int_0^t \theta_s \dd S_s - \int_0^t C_s \dd s, \quad t \geq 0,
\end{equation}
is satisfied. We then denote by $\Pi^0_t := \tfrac{\theta^0_t S^0_t}{X_t}$ and $\Pi_t := \tfrac{\theta_t S_t}{X_t}$ the fraction of wealth invested in the riskless and risky asset at time $t$, respectively.\footnote{Strictly speaking, $\Pi^0_t$ and $\Pi^1_t$ are not defined for $X_t = 0$, but this does not matter. We can for example set $\Pi^0_t :=0 $ and $\Pi_t :=1$ for $X_t =0$.} Noting that $\Pi^0_t + \Pi_t = 1$ by \eqref{eq:wealth process}, it follows that $X$ satisfies the SDE
\begin{align}
\dd X_t &= \theta^0_t \dd S^0_t + \theta_t \dd S_t - C_t \dd t \label{eqn:wealth:theta} \\
&= X_t \Pi^0_t r \dd t + X_t \Pi_t (\mu \dd t + \sigma \dd W_t)  - C_t \dd t \notag \\
&= X_t \Pi_t  \sigma \dd W_t + \left( X_t (r + \Pi_t (\mu - r))  - C_t\right)\dd t,
\label{eqn:original wealth process}
\end{align}
subject to $X_0 = x$. This means that we can describe an admissible investment-consumption strategy for initial wealth $x > 0$ more succinctly by a pair $(\Pi,C)= (\Pi_t,C_t)_{t \geq 0}$ of progressively measurable processes, where $\Pi$ is real-valued and $C$ is nonnegative, such that the SDE \eqref{eqn:original wealth process} has a unique strong solution $X^{x, \Pi, C}$ that is $\as{\P}$ nonnegative. We denote the set of admissible investment-consumption strategies for $x > 0$ by $\cA(x)$. A consumption stream $C$ is called \emph{attainable} for initial wealth $x > 0$ if there exists an investment process $\Pi$ such that $(\Pi,C) \in \cA(x)$, and we denote the set of attainable consumption streams for $x  > 0$ by $\mathscr{C}(x)$.

The objective of the agent is to maximise the expected discounted utility of consumption over an infinite time horizon for a given initial wealth $x > 0$. To any attainable consumption stream $C \in \cC(x)$, they associate a value $J(C) \in [-\infty, \infty]$, where
\begin{equation}
\label{eq:Jdef}
J(C) := \E \left[\int_0^\infty e^{-\delta t} U\left(C_t\right) \dd t \right].
\end{equation}
Here, $\delta \in \RR$ can be seen as a discount or impatience parameter; {see Appendix \ref{sec:numeraire} for a discussion on the economic interpretation of $\delta$, which also explains why, unlike much of the literature, we include the possibility $\delta \leq 0$.}
%\footnote{\cblue Note that unlike much of the literature we do not assume that $\delta>0$. One reason for this is that a deterministic change of accounting units leads to the same Merton consumption/investment optimisation problem, but under a different set of parameters, including a modified value of $\delta$. Then the problem with parameters $R,r,\mu,\sigma,\delta$ is equivalent to the problem with parameters $R,r+\gamma,\mu+\gamma,\sigma,\delta - (R-1)\gamma$ where $\gamma$ is an arbitrary constant. In particular, it is possible to make a perfectly a reasonable change of units, and change the sign of $\delta$. A second reason is that when $R>1$ the interpretation of the restriction $\delta>0$ is less clear. Note that $e^{-\delta t}\frac{C_t^{1-R}}{1-R} = \frac{(e^{-\delta t/(1-R)} C_t)^{1-R}}{1-R}$. Then when $R>1$, the discounted utility of consumption is equivalent to the utility of {\em upcounted} consumption. Equally plausible might be to take the utility of discounted consumption which would correspond to upcounting the utility of consumption. In the absence of a unequivocal rationale for a sign restriction on $\delta$, we allow $\delta$ to be unrestricted.}.
We assume that the agent has constant relative risk aversion (CRRA) or equivalently that $U: [0, \infty) \to [-\infty, \infty)$ takes the form $U(c) = \frac{c^{1-R}}{1-R}$, where $R \in (0,\infty) \setminus \{1 \}$ is the coefficient of relative risk aversion;\footnote{We follow the convention that $0^{1-R} := \infty$ for $R > 1$.} {$R=1$ is the case of logarithmic utility $U(c)= \log (c)$ and is discussed in Appendix \ref{sec:log}.}
%{\cblue Note, however, that there is an extra complication, in that $U$ can take either sign, and in principle (and in practice) we can have $\EX{\int_0^\infty e^{-\delta t} U(C_t)^+ dt} = \infty = \EX{\int_0^\infty e^{-\delta t} U(C_t)^- dt}$. In that case we need to decide what value to give to $J(C)$.}}.
Note that since $R \neq 1$, the sign of $U(c)$ is uniquely determined. Thus, if $\int_0^\infty e^{-\delta t} U\left(C_t\right) \dd t$ is not integrable, we can define $J(C) :=+\infty$ when $R<1$ and $J(C):=-\infty$ when $R>1$.

In summary, the problem facing the agent is to determine
\begin{equation}\label{eqn: value function definition power law case}
    V(x) \coloneqq \sup_{C \in \mathscr{C}(x)} J(C) %= \sup_{C \in \mathscr{C}(x)}\E \left[\int_0^\infty e^{-\delta t} u\left(c_t\right) dt \right]
    = \sup_{C \in \mathscr{C}(x)}\E \left[ \int_0^\infty e^{-\delta t} \frac{C_t^{1-R}}{1-R} \dd t \right].
\end{equation}

\section{The candidate value function}\label{sec:candidate}
From the homogeneous structure of the problem we expect (see for example, Rogers~\cite[Proposition 1.2]{rogers2013optimal}) that $V(\kappa x) = \kappa^{1-R}V(x)$ and that if $(\hat{\Pi},\hat{C})$ is an optimal strategy in $\mathscr{A}(x)$ then $(\hat{\Pi}_\kappa =\hat{\Pi},\hat{C}_\kappa = \kappa \hat{C})$ is optimal in $\mathscr{A}(\kappa x)$ for $\kappa > 0$. For this reason, we may guess that it is optimal to invest a constant fraction of wealth in the risky asset, and to consume a constant fraction of wealth. (Of course, this will be verified later.) So, consider an investment-consumption strategy that at each time $t$, invests a constant proportion of wealth $\Pi_t = \pi$ into the risky asset and consumes a constant fraction $\xi > 0$ of wealth per unit time, i.e., $ C_t  = \xi X_t$.%\footnote{This is well defined by Remark \ref{rem:parametrisation}.}
	
	 Then the agent's wealth process $X = X^{x, \pi,\xi X}$ is given by
\begin{equation}
  X_t =  x\exp\left(\pi \sigma W_t + \left( r + \pi (\mu - r)  - \xi - \frac{\pi^2 \sigma^2}{2}\right)t \right).
\end{equation}
Denoting the market price of risk or \emph{Sharpe ratio} by $\lambda := \frac{\mu -r}{\sigma}$, we obtain
\begin{equation}
\label{eq:cD}
 \frac{C_t^{1-R}}{1-R} = \frac{(\xi X_t)^{1-R}}{1-R} = \ \frac{x^{1-R} \xi^{1-R}}{1-R} \exp\left(\pi \sigma (1-R) W_t + (1-R)\left(r + \lambda \sigma \pi - \xi - \frac{\pi^2\sigma^2}{2}\right)t\right).
\end{equation}
Multiplying this by $e^{-\delta t}$ and taking expectations gives
\begin{equation}
\label{eq:c:exp}
\E\left[e^{-\delta t}\frac{C_t^{1-R}}{1-R}\right] = \  x^{1-R}\frac{\xi^{1-R}}{1-R} e^{-F(\pi, \xi) t},
\end{equation}
where
\begin{equation}
\label{eq:D pi xi}
F(\pi, \xi) = F(\pi,\xi; R, \delta, \lambda, r, \sigma):=\delta - (1-R)\left(r + \lambda \sigma \pi - \frac{\pi^2\sigma^2 }{2}R - \xi \right).
\end{equation}

{Provided that $F(\pi, \xi)>0$,
%\begin{equation}\label{eqn: beta sufficiently large for strategy (pi,c)}
%    \delta - (1-R)\left(r + \pi(\mu - r) - \frac{\pi^2\sigma^2 }{2}R - \xi \right) > 0,
%\end{equation}
we find that
\begin{equation}
\label{eq:J xi}
J(\xi X) =    \E \left[\int_0^\infty e^{-\delta t} \frac{\xi^{1-R} X_t^{1-R}}{1-R} \dd t \right] = \frac{x^{1-R}}{1-R}\frac{\xi^{1-R}}{F(\pi, \xi)}.
\end{equation}
We want to maximise this expression considered as a function of $\pi$ and $\xi$, where the maximisation is restricted to pairs $(\xi,\pi)$ for which $F(\pi,\xi)>0$.

Let $\eta$ be defined by
\begin{equation}\label{eqn: definition of eta}
\eta \coloneqq \frac{1}{R} \left[ \delta - (1-R)\left(r + \frac{\lambda^2}{2R}\right) \right],
\end{equation}
and suppose $\eta >0$. Set $\hat{\pi} = \frac{\lambda}{\sigma R}$ and $\hat{\xi} = \eta$.
Then it is easily seen that the right-hand-side of \eqref{eq:J xi} has a turning point at $(\pi,\xi) = (\hat{\pi},\hat{\xi})$, and that this turning point is in the region where $F(\pi,\xi)>0$ and gives the maximum in \eqref{eq:J xi}.}

%In order to maximise this over $\pi$, we want to minimise $(1-R)F(\pi, \xi)$, which is equivalent to maximising $\lambda \sigma \pi - \frac{\pi^2\sigma^2}{2}R$. This is achieved at
%$ %\begin{equation}
    %\hat{\pi} = \frac{\lambda}{\sigma R}.
%$ %\end{equation}
%In this case $\lambda \sigma \hat{\pi} - \frac{\hat{\pi}^2\sigma^2}{2}R = \frac{\lambda^2}{2R}$ and the problem then becomes to maximise
%\begin{equation}\label{eqn: maximisation problem for c case of R neq 1}
%    \frac{x^{1-R}}{1-R}\frac{\xi^{1-R}}{F(\hat \pi, \xi)} = \frac{x^{1-R}}{1-R} \frac{\xi^{1-R}}{\left(\delta - (1-R)(r + \frac{\lambda^2}{2R} - \xi) \right)}
%\end{equation}
%over $\xi$. A simple calculation shows that the maximum is attained at $\hat{\xi} = \eta$, where
%\begin{equation}\label{eqn: definition of eta}
%\eta \coloneqq \frac{1}{R} \left[ \delta - (1-R)\left(r + \frac{\lambda^2}{2R}\right) \right],
%\end{equation}
%provided that $\eta > 0$. (If $\eta \leq 0$, no maximum exists.)

Therefore, when $\eta>0$, the agent's optimal behaviour (at least over constant proportional strategies) and corresponding value function are given by
\begin{equation}\label{eqn: candidate value function power law case}
    \hat{\pi} = \frac{\mu - r}{\sigma^2 R}, \qquad \hat{\xi} = \eta, \qquad \hat{V}(x) \coloneqq J(\hat{\xi} X) = \frac{\eta^{-R}x^{1-R}}{1-R}.
\end{equation}

When $\eta \leq 0$, the problem is ill-posed. Indeed, if $R < 1$, then $F(\hat \pi, \xi) \downarrow 0$ as $\xi \downarrow - \frac{\eta R}{(1-R)}$ and hence $J(\xi X) \uparrow \infty$ by  \eqref{eq:J xi}. If $R > 1$, then $F(\pi, \xi) \leq F(\hat \pi,\xi) = R \eta + (1-R) \xi \leq R \eta \leq 0$ for every  $\pi \in \RR$ and $\xi \geq 0$. Hence, at least for constant proportional strategies $J(\xi X) = -\infty$. We will see in Corollary \ref{cor:ill posed} that $J(C) = -\infty$ for \emph{every} admissible consumption stream $C \in \cC(x)$.

\section{The verification argument under fiat conditions}
\label{sec:fiat}

In this section, we prove that our candidate optimal strategy $(\hat \pi, \hat \xi X)$ from \eqref{eqn: candidate value function power law case} is optimal in a subset of the class of all admissible strategies. Since the conditions defining that class are chosen precisely in such a way that the proof works, we call them \emph{fiat} conditions.

\begin{defn}
Fix $x > 0$. An investment-consumption strategy $(\Pi, C) \in \cA(x)$ is called \emph{fiat admissible} if the following three conditions are satisfied:
\begin{enumerate}
\item [(P)] The wealth process $X^{x,\Pi, C}$ is $\as{\P}$ positive.
\item [(M)] The local martingale $\int_0^\cdot e^{-\delta t} \sigma \Pi_t (X^{x,\Pi, C}_t)^{1-R}\dd W_t$ is a supermartingale.
\item [(T)] The transversality condition $\liminf_{t \to \infty} \E[e^{-\delta t} \frac{(X^{x,\Pi, C}_t)^{1-R}}{1-R}] \geq 0$ is satisfied.
\end{enumerate}
We denote the set of all fiat admissible investment-consumption strategies for $x > 0$ by $\cA^*(x)$. A consumption stream $C \in \cC(x)$ is called \emph{fiat attainable} for $x > 0$ if there is an investment process $\Pi$ such that $(\Pi, C) \in \cA^*(x)$. We denote the set of \emph{fiat attainable} consumption streams by $\cC^*(x)$.
\end{defn}

\begin{remark}
\label{rem:fiat}
As far as we are aware, the above notion of fiat admissible strategies has not been explicitly used in the literature before. However, the conditions (P), (M) and (T) or stronger versions thereof have been used explicitly or implicitly throughout the stochastic control literature on the Merton problem:
	
\begin{enumerate}
\item Condition (P) is (implicitly) assumed throughout most of the stochastic control literature dealing with the Merton problem;  a notable exception is \cite{karatzas1986explicit}. However, for $R > 1$, (P) can be assumed without loss of generality because any admissible strategy $(\Pi, C) \in \cA(x)$ violating (P) has $J(C) = -\infty$.
\item Condition (M) is implied by the stronger condition
\begin{itemize}
	\item [(M1)] The local martingale $\int_0^\cdot  \sigma \Pi_t (X^{x,\Pi, C}_t)^{1-R} e^{-\delta t}\dd W_t$ is a martingale.
\end{itemize}
It is not difficult to check that for $R < 1$, (M1) is implied by the even stronger condition
\begin{itemize}
	\item [(B)] $\Pi$ is uniformly bounded.
\end{itemize}
A common approach in the stochastic control literature is to assume (B), see e.g.~Davis and Norman~\cite[Equation (2.1)(B)]{davis1990portfolio}, Fleming and Soner~\cite[Equation IV.5.2]{fleming2006controlled}, or Pham~\cite[Equation (3.2)]{pham:06}, and then prove (M1) for $R < 1$.\footnote{Davis and Norman \cite[Proof of Theorem 2.1]{davis1990portfolio} argue that (B) implies (M1) also in the case $R > 1$ but this is not the case. See Example~\ref{eg:dncountereg}.}
\item Condition (T) is implied by the stronger standard transversality condition\footnote{Note, however, that if $R > 1$, (T) and (T1) are equivalent.}
\begin{itemize}
	\item [(T1)] $\lim_{t \to \infty} \E[e^{-\delta t} \frac{(X^{x,\Pi, C}_t)^{1-R}}{1-R}] = 0$.
\end{itemize}
When $R<1$, Davis and Norman~\cite[page 682]{davis1990portfolio} prove that (T1) is satisfied for any admissible strategy satisfying (B). Pham~\cite[Equation (3.39)]{pham:06} and Fleming and Soner~\cite[Equation IV.5.11]{fleming2006controlled} require (T1), and prove that the candidate optimal strategy has this property.
\end{enumerate}
\end{remark}

It is clear that $\cC^*(x) \subset \cC(x)$. The following result shows that the candidate optimal strategy $(\hat \pi, \hat \xi X)$ from \eqref{eqn: candidate value function power law case} is optimal in the class of fiat admissible strategies.
\begin{thm}
\label{thm:fiat}
Suppose $ \eta \coloneqq \frac{1}{R} [ \delta - (1-R)(r + \frac{\lambda^2}{2R})] > 0$. Let the function $\hat{V}: (0, \infty) \to \RR$ be given by $\hat V(x) = \frac{x^{1-R}}{1- R} \eta^{-R}$. Then for $x > 0$,
\begin{align}
    V^*(x) &:= \sup_{C \in \cC^*(x)}J(C) = J(\hat{C}) = \hat{V}(x),
\shortintertext{where the corresponding optimal investment-consumption strategy is given by $(\Pi,C)= (\hat{\Pi},\hat{C})$, where}
    \hat{\Pi} &=\frac{\lambda}{\sigma R}, \quad \hat{C} = \eta  X^{x,\hat \Pi, \hat C}.
    \label{eq:cand}
\end{align}
\end{thm}
\begin{proof}
First, we show that $V^*(x) \geq \hat{V}(x) = J(\hat C)$. By the arguments in Section \ref{sec:candidate}, it only remains to show that $\hat C$ is fiat attainable.
It follows from the construction of $\hat C$, that the wealth process $X^{x, \hat \Pi, \hat C}$ is $\as{\P}$ positive. Next, a similar calculation as in \eqref{eq:c:exp} shows that for each $T > 0$,
\begin{equation}
\E\left[\int_0^T e^{-2 \delta t} \sigma^2 \hat \pi^2 \left(X^{x, \hat \Pi, \hat C}_t\right)^{2-2R} \dd t\right] = \sigma^2 \hat \pi^2 \int_0^T \exp\left(\left(\frac{\lambda^2 (1-R)^2}{R^2} - 2\eta\right)t \right)\dd t < \infty.
\end{equation}
This implies that the local martingale $\int_0^\cdot \exp(-\delta t) \sigma \hat{\Pi}_t (X^{x,\hat{\Pi}, C}_t)^{1-R}\dd W_t$ is a (square-integrable) martingale and hence a supermartingale. Finally, \eqref{eq:c:exp} together with the fact that $F(\hat \pi, \eta) = \eta > 0$, implies that $(\hat \Pi, \hat{C})$ satisfies the transversality condition (T1).

Next, we show that $V^*(x) \leq \hat{V}(x)$. Let $(\Pi,C) \in \cA^*(x)$ be arbitrary. If $R > 1$, we may in addition assume without loss of generality that $C^{1-R}$ is integrable with respect to the identity process; for otherwise $J(C) = -\infty$. It suffices to argue that $J(C) \leq \hat V(x)$.

Set $X := X^{x,\Pi,C}$ for brevity and define the process $M = (M_t)_{t \geq 0}$ by
\begin{equation}
\label{eq:HJB0}
M_t = \int_0^t e^{- \delta s} U(C_s) \dd s + e^{-\delta t}{\hat V}(X_t).
\end{equation}
We want to apply It\^{o}'s formula to $M$. This is indeed possible as $\hat V$ is in $C^2(0, \infty)$ and $X$ is positive by fiat admissibility of $(\Pi, C)$. Note that ${\hat V}_x(X_t)$ is positive and ${\hat V}_{xx}(X_t)$ is negative.
Then, noting that the argument of $\hat V$ and its derivatives is $X_t$ throughout, we obtain
\begin{eqnarray*}
d M_t & = & \sigma \Pi_t X_t e^{-\delta t} {\hat V}_x \dd W_t +  e^{-\delta t}\left[ \frac{C_t^{1-R}}{1-R}  - \delta {\hat V} + (X_t(r + \sigma \lambda \Pi_t) - C_t){\hat V}_x  + \frac{\sigma^2}{2} \Pi_t^2 X_t^2 {\hat V}_{xx} \right]  \dd t \\
& = & \diff N_t + e^{-\delta t} L(\Pi_t,C_t; X_t, {\hat V}) \dd t.
\end{eqnarray*}
where $N_t = \int_0^t \sigma \Pi_s X_s e^{-\delta s} {\hat V}_x  \dd W_s= \int_0^t \eta^{-R} \sigma \Pi_s X^{1-R}_s e^{-\delta s}  \dd W_s$ is a local martingale and
\begin{equation}
\label{eq:Ldef}
L(\pi,c; x, v = v(x))  =  \frac{c^{1-R}}{1-R}  - \delta v + (x(r + \sigma \lambda \pi) - c)v_x  + \frac{\sigma^2}{2} \pi^2 x^2 v_{xx}.
\end{equation}
Maximising \eqref{eq:Ldef} over $\pi \in \RR$ and $c \geq 0$, shows that the optimisers are attained at $\hat \pi = \tfrac{\lambda}{\sigma}\tfrac{-v_x}{x v_{xx}}$ and $\hat c = v_x^{-1/R}$. Plugging in $\hat{V}$ shows that $L(\hat c, \hat \pi; x, \hat V) = 0$, which implies that $\hat{V}$ solves the Hamilton-Jacobi-Bellman equation
%& = & \left[ \sigma \lambda \pi x v_x + \frac{\sigma^2}{2} \pi^2 x^2 v_{xx} + \frac{\lambda^2}{2} \frac{v_x^2}{v_{xx}} \right]
%+ \left[ \frac{c^{1-R}}{1-R}  -  c v_x  -  \frac{R}{1-R} (v_x)^{1-1/R} \right] \nonumber  \\
\begin{equation}
	\label{eq:hjb1}
\sup_{\pi \in \R, c \geq 0} L(\pi,c;x,v) = 0.
% - \delta v + r x v_x + \frac{R}{1-R} (v_x)^{1-1/R} - \frac{\lambda^2}{2} \frac{v_x^2}{v_{xx}} =0
\end{equation}
It follows that
\begin{equation}
\label{eq:ineq:M}
{M}_t \leq \hat V(x) + {N}_t, \quad t \geq 0.
\end{equation}
Taking expectations and using fiat admissibility of $(\Pi, C)$ to ensure that $N$ is a supermartingale, we find for each $t \geq 0$,
\begin{equation}
\EX{M_t} \leq  \EX{\hat V(x) +  N_t} \leq \hat V(x).
\end{equation}
Taking the limit as $t $ goes to infinity, and using the monotone convergence theorem as well as the transversality condition, we obtain
\begin{align}
J(C) &= \lim_{t \to \infty} \E \left[ \int_0^t e^{- \delta s} \frac{C_s^{1-R}}{1-R} ds \right]  = \lim_{t \to \infty}  \EX{ M_t - e^{-\delta t} \hat V(X^{x,\Pi, C}_t)} \notag \\ \label{eq:ineq:J(C) V hat}
&\leq \limsup_{t \to \infty} \EX{ M_t} - \liminf_{t \to \infty} \E\left[e^{-\delta t} \hat V(X^{x,\Pi, C}_t)\right] \leq \limsup_{t \to \infty} \EX{ M_t} \leq \hat V(x).
\end{align}
This establishes the claim.
\end{proof}

\begin{remark}
\label{rem:UI}
A close inspection of the proof of Theorem \ref{thm:fiat} shows that for the optimal strategy $(\hat \Pi, \hat C)$, the process $\hat M = (\hat M_t)_{t \geq 0}$ given by $\hat M_t := \int_0^t e^{-\delta s} U(\hat C_s) \dd s + e^{-\delta t}\hat V(X^{x, \hat \Pi, \hat C})$ is a uniformly integrable martingale. Indeed, in this case $\hat N$ is a martingale and $\hat M = \hat V(x) + \hat N$. Hence, $\hat M$ is a martingale. It is uniformly integrable because, by the transversality condition (T1) and monotone convergence, equation \eqref{eq:ineq:J(C) V hat} implies that $\hat M_t$ converges in $L^1$ to $\hat M_\infty := \int_0^\infty e^{-\delta s} U(\hat C_s) \dd s$.
\end{remark}

For $R < 1$, the above fiat verification theorem can be easily generalised to a general verification theorem.
\begin{cor}
\label{cor:R<1}
Suppose $R < 1$ and $\eta>0$. Then $V(x) = \hat{V}(x)$.
\end{cor}

\begin{proof}
It is sufficient to show that (P), (M) and (T) are satisfied for general strategies, or to find a way of bypassing the relevant part of the argument.
First, (T) is automatically satisfied by the fact that $X^{1-R}/(1-R)$ is nonnegative.
Next, $M$ is nonnegative and hence $N$ is bounded below by $-\hat V(x)$ by \eqref{eq:ineq:M}. Therefore, $N$ is always a supermartingale and (M) is automatically satisfied.

Finally, to avoid imposing (P), one has to refine the argument in Theorem \ref{thm:fiat} by a stopping argument. To wit, fix an admissible strategy $(\Pi, C) \in \cA(x)$. Then for $n \in N$, set $\tau_n := \inf\{t \geq 0: X^{x, \Pi, C} \leq \frac{1}{n} \}$ and let $\tau_\infty:= \lim_{n \to \infty} \tau_n$. Then it is not difficult to check that  $X_t = X_t^{x, \Pi, C} \geq 1/n > 0$ if $t \leq \tau_n$ and $X_t = 0 = C_t$ if $t \geq \tau_\infty$.\footnote{More precisely, we have $\int_{\tau_\infty}^\infty C_s \dd s = 0$ $\as{\P}$}
Moreover, for each $n$, we get
\begin{equation*}
\EX{M^{\tau_n}_t} \leq \EX{\hat V(x) + N^{\tau_n}_t} \leq \hat V(x).
\end{equation*}
Now first taking the limit $t \to \infty$, we obtain
\begin{equation}
\E \left[ \int_0^{\tau_n} e^{- \delta s} \frac{C_s^{1-R}}{1-R} ds \right] \leq \limsup_{t \to \infty} \EX{M^{\tau_n}_t} \leq \hat V(x).
\end{equation}
Next, taking the limit $n \to \infty$, the result follows from the monotone convergence theorem and the fact that $\int_{\tau_\infty}^\infty C_s \dd s = 0$ $\as{\P}$
\end{proof}

\begin{remark}\label{rem:R<1}
The above approach of avoiding (P) is taken in \cite[Theorem 4.1]{karatzas1986explicit}. Note, however, that there the stopping argument is slightly more involved as it also requires stopping when the wealth process $X^{x, \Pi, C}$ or the quadratic variation of $\int_0^{\cdot} \sigma \Pi \dd W$ gets too large. But this additional stopping rather obfuscates the argument.

\end{remark}

\begin{remark}
	\label{rem:R>1}
If $R > 1$, extending Theorem \ref{thm:fiat} to general admissible strategies is far more involved. While condition (P) can be assumed without loss of generality (recall Part 1 of Remark~\ref{rem:fiat}), condition (M) is in general not satisfied as there are investment strategies $\Pi$ and consumption strategies $C$ such that $N$ fails to be a supermartingale, see Appendix~\ref{sec:egnonintegrable}. Note that these strategies are suboptimal because $L(\Pi_t,C_t;X_t,\hat{V})$ is (very) negative.
Finally, we have no reason to expect that the transversality condition (T) is satisfied. Indeed, (T) even fails for constant proportional strategies: If $\xi > \frac{\eta R}{R-1}$, then $F(\hat \pi,\xi) < 0$, and it follows from \eqref{eq:c:exp} that $\lim_{t \to \infty} \EX{\frac{e^{-\delta t}}{1-R} X^{x,\hat \pi, \xi X}_t} = -\infty$.
\end{remark}

\section{The general verification argument}
\label{sec:veri}

In this section, we present our general verification argument.
%\begin{comment}As already indicated in Remark~\ref{rem:David Norman}, \end{comment}
It is inspired by the perturbation argument of Davis and Norman, {see Appendix~\ref{ssec:dn}}. The key idea is to use the \emph{candidate optimal consumption strategy} as a \emph{stochastic perturbation} of the utility function. This yields a very elegant and simple argument that has the trio of advantages that it is no more difficult than the fiat verification argument in Theorem \ref{thm:fiat}, it does not need to distinguish between the case $R > 1$ and $R < 1$ and it does not involve any stopping argument.

%\footnote{\textit{Mutatis mutandis} the same argument gives also a full verification argument for the log case, i.e., $R = 1$.}

The following theorem contains the solution to the stochastically perturbed Merton problem. The subsequent corollary then lets this perturbation disappear.
Recall the notations of Theorem~\ref{thm:fiat}: $\eta = \frac{1}{R} [ \delta - (1-R)(r + \frac{\lambda^2}{2R})]$, $\hat \Pi = \frac{\lambda}{\sigma R}$ and
$\hat V(x) = \frac{x^{1-R}}{1-R} \eta^{-R}$.

\begin{thm}
\label{thm:veri}
Suppose $\eta > 0$. Denote by $Y = (Y_t)_{t \geq 0}$ the candidate optimal wealth process started from unit initial wealth $1$, i.e., $Y_t := X_t^{1,\hat \Pi, \eta X}$, and by $G = (G_t)_{t \geq 0}$, the corresponding optimal consumption stream, i.e., $G_t = \eta Y_t$.
Fix $\epsilon > 0$, define the function $U_\epsilon: [0, \infty) \times (0, \infty) \to (-\infty, \infty)$ by $U_\epsilon(c, g) = \frac{(c + \epsilon g)^{1-R}}{1-R}$, and for
 an attainable consumption stream $C$ consider
\[ J_\epsilon(C) := \E \left[ \int_0^\infty e^{-\delta t} U_\epsilon(C_t, G_t) \dd t \right] = J(C + \epsilon G). \]
Then for $x > 0$,
\begin{equation}
V_\epsilon(x) :=  \sup_{C \in \mathscr{C}(x)} J_\epsilon(C) = \hat V(x + \epsilon).
\end{equation}
Moreover, the supremum is attained when $\Pi = \hat{\Pi}$ and $C= \hat{C}$ where $\hat{C} = \eta X^{x, \hat \Pi, \hat C}$.
\end{thm}

\begin{proof}
First, from the SDE for the wealth process \eqref{eqn:original wealth process} we have that $X^{x, \hat \Pi, \eta X} + \epsilon Y = X^{x + \epsilon, \hat \Pi, \eta X}$. It follows that
$\hat C + \epsilon G = \eta X^{x + \epsilon, \hat \Pi, \eta X} \in \cC(x+\epsilon)$, which together with Theorem \ref{thm:fiat} implies that $J_\epsilon(\hat C) = J(\hat C + \epsilon G) =  \hat V(x+\epsilon)$.
	
It remains to show that $V_\epsilon(x) \leq \hat V(x + \epsilon)$. The argument is very similar to the one in the proof of Theorem  \ref{thm:fiat}. Let $(\Pi,C) \in \cA(x)$ be arbitrary and set $X := X^{x,\Pi, C}$ for brevity. The dynamics of $X + \epsilon Y$ are given by
\begin{equation}\label{eqn: equation for X_t + theta Y_t first use}
  \diff (X_t + \epsilon Y_t) = \left(\sigma \Pi_t X_t + \frac{\lambda}{R} \epsilon Y_t\right) \dd W_t + \left(X_t (r  +\Pi_t \sigma \lambda) - C_t + \left(r + \frac{\lambda^2}{R} - \eta \right)\epsilon Y_t \right) \dd t.
\end{equation}
Define the process  $M^\epsilon = (M^\epsilon_t)_{t \geq 0}$ by
\[ M^\epsilon_t = \int_0^t e^{-\delta s} U_\epsilon(C_s, G_s) \dd s + e^{-\delta t} \hat{V}(X_t + \epsilon Y_t). \]
{We proceed to apply It\^o's formula to $M^\epsilon$. %Adding and subtracting $\frac{R}{1-R} (\hat{V}_x)^{1-1/R} + \epsilon \eta Y_t \hat{V}_x$ and $\frac{\lambda^2}{2} \frac{ V_x^2}{V_{xx}} + \frac{\lambda^2}{R} \epsilon Y_t \hat{V}_x$ and
Noting that the argument of $\hat{V}$ and its derivatives is $(X_t + \epsilon Y_t)$ throughout, we obtain
\begin{eqnarray*}
 \diff M^\epsilon_t & = & e^{-\delta t} \frac{(C_t + \epsilon \eta Y_t)^{1-R}}{1-R} \dd t  + e^{-\delta t} \! \left[-\delta \hat{V} \dd t +  \hat{V}_x \dd (X_t + \epsilon Y_t) + \frac{1}{2} \hat{V}_{xx} \dd [X + \epsilon Y]_t \right]\\
& = & \diff N^\epsilon_t + e^{-\delta t} L^\epsilon(\Pi_t,C_t;X_t,Y_t,\hat{V}) \dd t
\end{eqnarray*}
where $N^\epsilon_t = \int_0^t e^{-\delta s} \eta^{-R} (X_s + \epsilon Y_s)^{-R}(\sigma \Pi_s X_s + \frac{\lambda \epsilon}{R}Y_s)\dd W_s$ and, with $z = x + \epsilon y$,
\begin{eqnarray*}
\lefteqn {L^\epsilon(\pi,c;x,y,v=v(z)) } \\
& = & \frac{(c + \epsilon \eta y)^{1-R}}{1-R} - \delta v + \left[ x(r + \pi \sigma \lambda) - c + (r + \frac{\lambda^2}{2}-\eta) \epsilon y \right] v_z
           + \frac{1}{2}\left(\sigma \pi x + \frac{\lambda \epsilon y}{R} \right)^2 v_{zz} \\
& = & L \left( \frac{\pi x}{z} + \frac{\lambda \epsilon y}{\sigma Rz}, c + \epsilon \eta y; z, v=v(z) \right).
\end{eqnarray*}
Here $L$ is the operator defined in \eqref{eq:Ldef}. Then\footnote{The inequality is in fact an equality since the maximum over $\tilde{c}$ is attained at ${\hat V}^{-1/R}(z)=  \eta z =\eta(x+\epsilon y) \geq \epsilon \eta y$.}
\begin{align*}
&\sup_{\pi \in \R, c \geq 0} L^\epsilon(\pi,c;x,y,\hat{V}=\hat{V}(z))  =
\sup_{\pi \in \R, c \geq 0} L \left( \frac{\pi x}{z} + \frac{\lambda \epsilon y}{\sigma Rz}, c + \epsilon \eta y; z, \hat{V}=\hat{V}(z) \right) \\
&\qquad = \sup_{\tilde{\pi} \in \R, \tilde{c} \geq \epsilon \eta y} L \left( \tilde{\pi}, \tilde{c}, z, \hat{V}={\hat V}(z) \right)
\leq \sup_{\tilde{\pi} \in \R, \tilde{c} \geq 0} L \left( \tilde{\pi}, \tilde{c}; z, \hat{V}=\hat{V}(z) \right) = 0
\end{align*}
where the final equality follows from \eqref{eq:hjb1}.
\begin{comment}
\begin{eqnarray}\label{eqn: first dynamics of M_t power law case}
    \diff M^\epsilon_t \! & = & e^{-\delta t} \frac{(C_t + \epsilon \eta Y_t)^{1-R}}{1-R} \dd t  + e^{-\delta t} \! \left[-\delta \hat{V} \dd t +  \hat{V}_x \dd (X_t + \epsilon Y_t) + \frac{1}{2} \hat{V}_{xx} \dd [X + \epsilon Y]_t \right]\\
    & = &  e^{-\delta t}\hat{V}_x \left(\sigma \Pi_t X_t + \frac{\lambda}{R} \epsilon Y_t\right) \dd W_t  \\
        &&+ e^{-\delta t} \left[ \frac{(C_t +  \epsilon \eta Y_t)^{1-R}}{1-R} - (C_t + \epsilon \eta Y_t) \hat{V}_x - \frac{R}{1-R} (\hat{V}_x)^{1-1/R} \right] \dd t \\
&& + e^{-\delta t} \left[ \lambda \left( \sigma \Pi_t X_t + \frac{\lambda}{R} \epsilon Y_t \right) \hat{V}_x + \frac{1}{2} \left( \sigma \Pi_t X_t + \frac{\lambda}{R} \epsilon Y_t \right)^2 \hat{V}_{xx} + \frac{\lambda^2}{2} \frac{ \hat{V}_x^2}{\hat{V}_{xx}}\right] dt \\
 && + e^{-\delta t} \left[ - \delta \hat{V} + r(X_t+ \epsilon Y_t) \hat{V}_x +\frac{R}{1-R} (\hat{V}_x)^{1-1/R} - \frac{\lambda^2}{2} \frac{ \hat{V}_x^2}{\hat{V}_{xx}} \right. \\
&& \hspace{50mm} \left. + \left( \frac{\lambda^2}{R} - \eta \right) \epsilon Y_t \hat{V}_x + \epsilon \eta Y_t \hat{V}_x - \frac{\lambda^2}{R}\epsilon Y_t  \hat{V}_x \right] \dd t\\
    & =: & \dd N^\epsilon_t + A^{1,\epsilon}_t \dd t + A^{2,\epsilon}_t \dd t +  A^{3,\epsilon}_t \dd t.
\end{eqnarray}
\end{comment}
This gives
\begin{equation}
\label{eq:Meps:ineq}
M^\epsilon_t \leq \hat V(x + \epsilon) + N^\epsilon_t, \quad t \geq 0.
\end{equation}
} %\end{\cblue}
Next, define the process $\Lambda^\epsilon = (\Lambda^\epsilon_t)_{t \geq 0}$ by
\begin{equation}
\Lambda^\epsilon_t := \int_0^t e^{-\delta s} U_\epsilon(0, G_s) \dd s + e^{-\delta t} \hat{V}(0 + \epsilon Y_t) = \int_0^t e^{-\delta s} U(\epsilon G_s) ds + e^{-\delta t} \hat{V}(\epsilon Y_t).
\end{equation}
Then $\Lambda^\epsilon \leq M^\epsilon$ by monotonicity of $U$ and $\hat{V}$. Using that $\Lambda^\epsilon$ is a (UI) martingale by Remark \ref{rem:UI}, it follows that $N^\epsilon$ is bounded below by the (UI) martingale $-\hat V(x+\epsilon) - \Lambda^\epsilon$ and hence a supermartingale.

Taking expectation in \eqref{eq:Meps:ineq}, we find for each $t \geq 0$,
\begin{equation}
\label{eq:Meps:mean ineq}
\EX{M^\epsilon_t} \leq \EX{\hat V(x + \epsilon) + N^\epsilon_t} \leq \hat V(x + \epsilon).
\end{equation}

Next, note that $X + \epsilon Y$ satisfies the transversality condition (T) since
\begin{equation}
\label{eq:T for X+epsY}
\liminf_{t \to \infty} \EX{e^{-\delta t} \frac{(X_t + \epsilon Y_t)^{1-R}}{1-R}} \geq  \epsilon^{1-R} \liminf_{t \to \infty}\EX{e^{-\delta t}  \frac{Y_t^{1-R}}{1-R}} = 0.
\end{equation}

Taking the limit in \eqref{eq:Meps:mean ineq} as $t$ goes to infinity and using \eqref{eq:T for X+epsY}, we may conclude that for any $C \in \mathscr{C}(x)$,
\begin{align}
J_\epsilon(C) &= \lim_{t \to \infty} \E \left[ \int_0^t e^{- \delta s} \frac{(C_s + \epsilon G_s)^{1-R}}{1-R} ds \right]  = \lim_{t \to \infty}  \EX{M^\epsilon_t - e^{-\delta t} \hat V(X_t+ \epsilon Y_t)} \notag \\
&\leq \limsup_{t \to \infty} \EX{M^\epsilon_t} - \liminf_{t \to \infty} \EX{e^{-\delta t} \eta^{-R} \frac{(X_t + \epsilon Y_t)^{1-R}}{1-R}} \\
&\leq \limsup_{t \to \infty} \EX{M^\epsilon_t} \leq \hat V(x+\epsilon). \qedhere
\end{align}

\end{proof}

\begin{remark}
The perturbation of the problem by the additional consumption of $\epsilon G$ elegantly and simply transforms the problem to one in which the fiat conditions (P), (M) and (T) are satisfied. Since $Y$ is positive $\as{\P}$, the same is trivially true for $X+ \epsilon Y$. Moreover, $J(\epsilon G) = \epsilon^{1-R} J(G) > -\infty$ and this allows us to easily find an integrable lower bound on $N^{\epsilon}$ and hence conclude it is a supermartingale. Again $Y$ satisfies a transversality condition (T) and so the same is trivially true for $X + \epsilon Y$.
\end{remark}

\begin{remark}
One interpretation of the theorem is that a financially-savvy benefactor gives the agent an additional consumption stream based on an initial wealth $\epsilon$ which is invested optimally by the benefactor. Then, if the agent behaves optimally with their own wealth, the two consumption streams and investment strategies remain perfectly aligned to each other, and the derivation and valuation of the candidate optimal strategy is simple and immediate.
\end{remark}

\begin{cor}
\label{cor:veri}
Suppose $ \eta > 0$. Then for $x > 0$,
\begin{align}
V(x) &:= \sup_{C \in \cC(x)}J(C) = J(\hat{C}) = \hat{V}(x).
%\label{eq:cand}
\end{align}
\end{cor}

\begin{proof}
The equality $J(\hat C) =  \hat{V}(x)$ follows from Theorem \ref{thm:fiat}. It remains to establish that $V(x) \leq \hat V(x)$.
Using the notation of Theorem \ref{thm:veri}, for any $C \in \mathscr{C}(x)$, we get $J(C) \leq J_\epsilon(C) \leq V_\epsilon(x) = \hat{V}(x + \epsilon)$. Letting $\epsilon \downarrow 0$, we conclude that $V(x) \leq \hat{V}(x)$.
\end{proof}

We finish this section by showing that in the case $R > 1$ if $\eta \leq 0$, \emph{every} $C \in \cC(x)$ has $J(C) = -\infty$.
\begin{cor}
\label{cor:ill posed}
Suppose that $R > 1$ and $ \eta \leq 0$. Then
\begin{equation}
V(x) = \sup_{C \in \cC(x)}J(C)  = -\infty.
\end{equation}
\end{cor}

\begin{proof}
Fix $C\in \cC(x)$. It suffices to show that $J(C) = -\infty$. We use an approximation argument. For $n \in \NN$ set $\delta_n := \delta + R(\frac{1}{n} - \eta)$. Then $\delta_n > \delta$ and $\eta_n := \frac{1}{R} [ \delta_n - (1-R)(r + \frac{\lambda^2}{2R})] = \frac{1}{n} > 0$.
Then using that $U(c) < 0$ for $c \geq 0$, it follows from Theorem \ref{thm:veri}
\begin{equation}
J(C) = \EX{\int_0^\infty e^{-\delta s} U(C_s) \dd s} \leq \EX{\int_0^\infty e^{-\delta_n s} U(C_s) \dd s} \leq \frac{x^{1-R}}{1-R} (\eta_n)^{-R} , \quad n \in \NN.
\end{equation}
Taking the limit on the right hand side as $n$ goes to $\infty$, it follows that $J(C) = -\infty$.
\end{proof}

\newpage

\bibliographystyle{plain}
\bibliography{merton}
\include{bibliography}

\appendix

\section{An example for which $N$ fails to be a supermartingale}
\label{sec:egnonintegrable}
For $R > 1$, the process $N$ in the proof Theorem \ref{thm:fiat} can fail to be supermartingale. We first give an abstract version of an example and then two concrete specifications.
\begin{example}
\label{eg:dncountereg}
Let $(\Pi, C) \in \cA(x)$ be such that $X = X^{x, \Pi, C}$ has $\as{\P}$ positive paths. Define the stopping time
\begin{equation}
\tau := \inf \left\{t \geq 0: \int_0^t \eta^{-R} \sigma \Pi_s X_s^{1-R} e^{-\delta s} \dd W_s = 1\right\}.
\end{equation}
If $\tau$ is bounded, then $N$ fails to be a supermartingale because $\EX{N_\tau} = 1 > 0 = \EX{N_0}$.

The above abstract situation can be achieved either by “wild” investment or by “too fast'' consumption, or a combination of the two.

For an example of a “wild” investment strategy $\Pi$, assume that $\mu \geq r >  0$ and define the stopping time
\begin{equation}
\tilde \tau := \inf \left\{t \geq 0: \int_0^t \frac{\eta^{-R} \sigma e^{-\delta s}}{1 - s} \dd W_s = 1\right\}.
\end{equation}
Note that $\tilde \tau < 1$ $\as{\P}$ since $\int_0^1 \left(\frac{\eta^{-R} \sigma e^{-\delta s}}{1 - s}\right)^2 \dd s = \infty$. Then define $(\Pi, C) \in \cA(x)$ by
\begin{equation*}
\Pi_t = \frac{1}{1 - t} X^{R-1}_t \1_{\{t \leq \tilde \tau\}}, \quad  C_t := r X_t + \Pi_t X_t (\mu - r).
\end{equation*}
Then the corresponding wealth process $X$ is a stopped and time changed CEV process:
\begin{equation}
\dd X_t = X^R_t \frac{\sigma}{1 - t}  \1_{\{t \leq \tilde \tau\}} \dd W_t, \quad X_0 = x.
\end{equation}
Since $R>1$, $X$ remains positive. Since $\tau = \tilde \tau$ $\as{\P}$ we have $\tau<1$ $\as{\P}$ and $N$ fails to be a supermartingale.

For an example of a “too fast” consumption strategy $C$ (with bounded investment strategy $\Pi$), assume that $\mu \geq r >  0$ and define the stopping time
\begin{equation}
\bar \tau := \inf \left\{t \geq 0: \int_0^t \frac{x^{1-R}\eta^{-R} \sigma  e^{\sigma (1-R) W_s - (\delta + \frac{(1-R)}{2} \sigma^2)s}}{1 - s} \dd W_s = 1\right\}.
\end{equation}
Note that $\bar \tau < 1$ $\as{\P}$ since $\int_0^1 (\frac{x^{1-R}\eta^{-R} \sigma  e^{\sigma (1-R) W_s - (\delta + \frac{(1-R)}{2} \sigma^2)s}}{1 - s})^2 \dd s$ $\as{\P}$ Then define $(\Pi, C) \in \cA(x)$ by
\begin{equation*}
\Pi_t = \1_{\{t \leq \bar \tau\}}, \quad  C_t := \frac{1}{R-1}\frac{X_t}{1-t} \1_{\{t \leq \bar \tau\}}+ r X_t + \Pi_t X_t (\mu - r).
\end{equation*}
Then the corresponding wealth process satisfies the SDE
\begin{equation}
\dd X_t = \sigma X_t \1_{\{t \leq \bar \tau\}} \dd W_t - \frac{1}{R-1} \frac{X_t}{1-t}\1_{\{t \leq \bar \tau\}} \dd t, \quad X_0 = x.
\end{equation}
It is not difficult to check that this has the solution
\begin{equation}
X_t = x (1-t \wedge \bar \tau )^{\frac{1}{R-1}}e^{\sigma W_{t \wedge \bar \tau} - \frac{1}{2} \sigma^2{(t \wedge \bar \tau)}}
\end{equation}
which is well-defined and positive by the fact that $\bar \tau < 1$ $\as{\P}$ Since $\tau = \bar \tau$ $\as{\P}$, we have $\tau<1$ $\as{\P}$ and $N$ fails to be a supermartingale.
\end{example}

\section{Verification approaches for $R > 1$}
\label{sec:literatureR>1}

As we have explained in Remark \ref{rem:R>1}, a verification argument for general admissible strategies requires some additional ideas for the case $R > 1$. In this section, we discuss the two most general approaches in the extant stochastic control literature. Both approaches first consider a perturbation of the problem (or the candidate solution) and then let the perturbation disappear.

\subsection{Perturbation with finite bankruptcy}
\label{ssec:KLSS}
The first perturbation approach is by Karatzas et al~\cite{karatzas1986explicit} who study an optimal investment-consumption problem with bankruptcy for a general utility function which is of interest in its own right, building on earlier work \cite{lehoczky1983optimal} by a subset of the authors. In the following, we only describe their contribution towards the solution of the Merton problem for CRRA utilities. We assume $R > 1$, and we use our notation.

Assume that $\delta > 0$ and $r > 0$. %\footnote{Karatzas et al~\cite{karatzas1986explicit} assume that $r >0$ but their argument (for the CARA case) extends to $r = 0$.}
For an admissible strategy $(\Pi, C) \in \cA(x)$, denote the \emph{bankruptcy time} $\tau_0 = \tau^{x,\Pi,C}_0 = \inf \{ t : X^{x,\Pi,C}_t = 0 \}$. Then choose a finite bankruptcy value ${P \in (-\infty, 0)}$ and consider the problem with bankruptcy:
\begin{equation}
\label{eq:KLSS}
V^P(x) := \sup_{C \in \mathscr{C}(x)} J^P(C) = \sup_{C \in \mathscr{C}(x)} \E \left[ \int_0^{\tau^{x,\Pi,C}_0} e^{-\delta t} U(C_t) \dd t + e^{-\delta \tau^{x,\Pi,C}_0} P \right].
\end{equation}
Note that the classical Merton problem corresponds to the limiting case $P = -\infty$.

Karatzas et al~\cite{karatzas1986explicit} show the following:
\begin{enumerate}
\item [(A)]Suppose that a $C^2$-function $\hat V^P: (0, \infty) \to (P, 0)$ solves the HJB equation corresponding to the optimisation problem \eqref{eq:KLSS} given by
\begin{align}
\label{eq:HJB:KLSS}
\delta \tilde V(x) =& \sup_{c\geq0,\pi}\left[\tilde V'(x)((\mu - r)\pi x + (rx - c)) + \frac{1}{2}\pi^2 \sigma^2 x^2 \tilde V''(x) + U(c) \right], \quad x>0.
\end{align}
subject to $\lim_{x \downarrow 0} \tilde V(x) = P$.

Then $\hat V^P(x) = V^P(x)$ for all $x \in [0, \infty)$.

\item [(B)] For each $P \in (-\infty,0)$, there exists a $C^2$-function $\hat V^P: (0, \infty) \to (P, \infty)$ that solves the HJB equation~\eqref{eq:HJB:KLSS} with $\lim_{x \downarrow 0}\hat V^P(x)=P$.

\item [(C)] $V(x) \leq \lim_{P \downarrow -\infty} \hat V^P(x) =  \hat V(x)$, which together with $\hat V(x) \leq V(x)$ establishes the claim.
\end{enumerate}
Here, the argument for (A) is relatively straightforward; see \cite[Theorem 4.1]{karatzas1986explicit} and not more difficult than the proof of our Theorem 4.3. Similarly, the argument for (C) is easy: the first inequality follows from the fact that $V(x) \leq V^P(x) \leq \hat V^P(x)$ for each $x > 0$ and $P \in \RR_-$ by the definition of $V^P$ and (A); the second inequality is straightforward using the explicit form for $\hat{V}^P$.

But the main difficulty -- and great ingenuity -- of the argument in \cite{karatzas1986explicit} is (B). Indeed, a direct calculation for $r > 0$ case takes at least two pages and yields the answer:%\footnote{For $r  = 0$, we get the simpler answer $\hat{V}(x,P) = \eta^{-1} (\eta x + (\eta(1-R)P)^{1/(1-R)})^{1-R}/(1-R).$}
\begin{equation}
\label{eq:VP KLSS}
\hat{V}^P(x) =  \frac{\nu}{\eta(R - \nu)} \left(\frac{\eta}{R}\frac{R - \nu}{1 - \nu}(1-R)P\right)^{\frac{1-\nu}{1-R}} (\hat{C}^P(x))^{\nu - R} + \eta^{-1}\frac{(\hat{C}^P(x))^{1-R}}{1-R},
\end{equation}
where the function $\hat{C}^P(x)$ describing the optimal consumption is the inverse of the function
\begin{equation*}
I^P(c) = - \eta^{-1} \left(\frac{\eta}{R}\frac{\nu-R}{\nu - 1}(1-R)P\right)^{\frac{1-\nu}{1-R}}c^{\nu} + \frac{c}{\eta},
\end{equation*}
and $\nu$ is the negative root of the equation $\frac{\lambda^2}{2} \zeta^2 + (r - \delta - \frac{\lambda^2}{2})R \zeta - rR^2  = 0$.
%We return to this approach in Section~\ref{ssec:SethiImpatience}.

\subsection{Perturbation of the value function}
\label{ssec:dn}
The second perturbation approach is by Davis and Norman~\cite{davis1990portfolio} who study the Merton problem with transaction costs; the perturbation argument for $R > 1$ in the frictionless case is a fortunate by-product, and not the main contribution of the paper. Again we will use our notation to describe their approach.

Assume that $\delta > 0$ and $r > 0$. Denote by $\cA^b(x)$ the set of all admissible strategies $(\Pi, C)$ for which $\Pi$ is uniformly bounded, write $\cC^b(x)$ for the corresponding set of attainable consumption strategies and set $V^b(x) := \sup_{C \in \cC^b(x)} J(C)$. For $\zeta > 0$, consider the perturbed value function $\hat{V}^{\zeta}(x) = \hat{V}(x+\zeta)$ and for $(\Pi, C) \in \cA^b(X)$ (such that $C^{1-R}$ is integrable with respect to the identity process), consider the process $M^{\zeta}$ defined by
 \begin{equation}
M^{\zeta} = \int_0^t e^{-\delta t} U(C_s) \dd s + e^{-\delta t} \hat V^\zeta(X_t).
 \end{equation}
Then the same argument as in the  proof of Theorem  \ref{thm:fiat} but with $\hat V$ replaced by $\hat V^\zeta$ yields
 \begin{equation}
 \dd M^\zeta_t = \dd N^\zeta_t + { L(\Pi_t, C_t; X_t, \hat{V}^\zeta) \dd t} %A^{1,\zeta}_t \dd t + A^{2,\zeta}_t \dd t + A^{3, \zeta}_t \dd t,
 \end{equation}
%where the only difference is that $\hat V$ and its derivatives are replaced by $\hat V^\zeta$ and its derivatives.
%Then as in the proof of Theorem  \ref{thm:fiat} {\cblue it follows that $L(\pi,c;x,\hat V^\zeta)  \leq  0$. Moreover,
{Using that ${\hat V}^\zeta(x) = {\hat V}(x + \zeta)$, it is straightforward to check that $\sup_{\pi \in \R,c \geq 0} L(\pi,c;x,{\hat V}^\zeta) = -r \zeta \hat V^\zeta_x(X_t)e^{-\delta t} \leq 0$. It follows that, under the crucial assumption that $r \geq 0$, we have $L(\Pi_t, C_t; X_t, \hat{V}^\zeta) \leq 0$.} Finally, using that $\Pi$ and $\hat V^\zeta_x$ are bounded, it is not difficult to check that $N^\zeta$ is a square integrable martingale. Now following the proof of Theorem  \ref{thm:fiat}, and using that $|\hat V^\zeta|$ is bounded and $\delta > 0$ it follows that
 \begin{equation*}
 J(C) \leq \limsup \EX{M^\zeta_t} - \liminf \EX{e^{-\delta t} \hat V^\zeta(X^{x, \Pi, C}_t)} \leq \limsup \EX{M^\zeta_t} \leq \hat V^\zeta(x).
 \end{equation*}
 We may conclude that $V^b(x) \leq \hat V^\zeta(x)$ and taking the limit as $\zeta \downarrow 0$, it follows that $V^b(x) = \hat V(x)$.

\section{The Merton problem with logarithmic utility}
\label{sec:log}
The case of logarithmic utility $U(c)= \log (c)$ corresponds to the case of unit coefficient of relative risk-aversion. The function $\log$ differs from the other CRRA utility functions in that it may take both signs but much of the analysis goes through in exactly the same way.

Under logarithmic utility, the problem facing the agent is to choose an admissible strategy $(\Pi,C) \in \mathscr{A}(x)$ so as to find
\begin{equation}\label{eqn: value function definition log case}
	V(x) \coloneqq \sup_{C\in\mathscr{C}(x)} J(C)\coloneqq \sup_{C\in\mathscr{C}(x)}\E \left[\int_0^\infty e^{-\delta t} \log\left(C_t\right) \dd t \right].
\end{equation}
Since $\log$ takes both positive and negative values, we make the definition\footnote{There are other ways
$\E \left[\int_0^\infty e^{-\delta t} \log\left(C_t\right) \dd t \right]$ might be defined. For example, one could make the subtly different definition
$\E \left[\int_0^\infty e^{-\delta t} \log\left(C_t\right) \dd t \right] := \E\left[ \int_0^\infty e^{-\delta t} \log\left(C_t\right)^+ \dd t \right] - \E\left[ \int_0^\infty e^{-\delta t} \log\left(C_t\right)^- \dd t \right]$, with the convention $\infty - \infty := -\infty$. The advantage of our definition is that in the case $\delta \leq 0$ it leads to a much cleaner final statement of results. In the case $\delta>0$, the two definitions are equivalent because then $\E\left[ \int_0^\infty e^{-\delta t} \log\left(C_t\right)^+ \dd t \right] < \infty$ for all admissible consumption streams $C$.}
$$\E \left[\int_0^\infty e^{-\delta t} \log\left(C_t\right) \dd t \right] := \E\left[\left(\int_0^\infty e^{-\delta t} \log\left(C_t\right) \dd t\right)^+ \right] - \E\left[\left(\int_0^\infty e^{-\delta t} \log\left(C_t\right) \dd t\right)^- \right],$$
where for each $\omega$,
$$\int_0^\infty e^{-\delta t} \log\left(C_t\right) \dd t  := \int_0^\infty e^{-\delta t} \log\left(C_t\right)^+ \dd t - \int_0^\infty e^{-\delta t} \log\left(C_t\right)^- \dd t,$$
with the standard convention that $\infty - \infty := -\infty$.

As before, we postulate a constant proportion of wealth for both our optimal investment and our optimal consumption, $\Pi_t = \pi$ and $C_t = \xi X_t$. In this case, our wealth process is given by \eqref{eq:wealth process}. Taking logarithms we find that
\begin{equation}
	\label{eq:log:const}
\log(C_t) =\log(\xi x) + \pi\sigma W_t + \left( r + \lambda\sigma\pi - \xi - \frac{\pi^2 \sigma^2}{2}\right)t.
\end{equation}
This implies that
\begin{equation}\label{eq:ln xi X EX}
	\E [e^{-\delta t} \log(\xi X_t)] = e^{-\delta t}\log(\xi x) + e^{-\delta t}\left( r + \lambda\sigma\pi  - \xi - \frac{\pi^2 \sigma^2}{2}\right)t.
\end{equation}

The well-posedness condition $\eta>0$ for $R=1$ is equivalent to $\delta>0$. So, suppose first that $\delta > 0$. Then
% \begin{align}
\begin{equation}	
J(\xi X) = %&~\E\left[\int_0^\infty e^{-\delta s} \left(r + \lambda\sigma\pi  - \xi - \frac{\pi^2 \sigma^2}{2} \right) s + e^{-\delta s}\log(\xi x)  ds\right] \\ =&~
\frac{1}{\delta^2}\left(\delta \log(\xi) + \delta \log(x) + \left(r + \lambda\sigma\pi  - \xi - \frac{\pi^2 \sigma^2}{2} \right) \right).
%\end{align}
\end{equation}
By taking derivatives with respect to $\pi$ and $\xi$, we find that this is maximised at
$\pi = \hat{\pi} := \frac{\lambda}{\sigma}$ and $\xi = \hat{\xi} := \delta$.
Note that this corresponds to the candidate optimal strategy given in \eqref{eq:cand} for $R=1$. %Let $\kappa = r + \frac{\lambda^2}{2}$.
Then, the candidate value function is given by
\begin{align}\label{eqn: definition of V hat log case}
	\hat{V}(x) \coloneqq& J \left(\hat{\xi} X\right) = \frac{1}{\delta^2}\left(\delta\log(\delta x) + r + \frac{\lambda^2}{2} - \delta\right).
\end{align}

To prove optimality, one considers the stochastically perturbed Merton problem corresponding to the aggregator $U_\epsilon(c,g) = \log (c+\epsilon g)$ and $G = \delta Y$ for $Y$ the wealth process under the candidate optimal strategy. The corresponding version of Theorem \ref{thm:veri} then goes through exactly as when $R\in(0,\infty) \setminus \{1 \}$.

When $\delta \leq 0$ the problem becomes delicate. Set $\kappa:= r + \frac{\lambda^2}{2}$.

If $\kappa > 0$ choose $\Pi_t := \hat \pi = \tfrac{\lambda}{\sigma}$ and $C_t := \tfrac{\kappa}{2} X_t$. Then \eqref{eq:log:const} gives
\begin{equation}
\log(C_t) = \log\left(\frac{\kappa}{2} x\right) + \lambda W_t + \frac{\kappa}{2} t.
\end{equation}
It follows from the strong law of large numbers, that for $\P$-a.e.~$\omega$, $\log(C_t(\omega)) \geq 1$ for all $t$ sufficiently large. Hence
\begin{equation*}
\int_0^\infty e^{-\delta t} \log\left(C_t\right) \dd t = +\infty\;\; \as{ \P}
\end{equation*}
whence $J(C) = +\infty$.

If $\kappa \leq 0$, we show that for \emph{every} admissible consumption stream
\begin{equation}
	\label{eq:log:kappa:neg}
	\int_0^\infty e^{-\delta t} \left(\log\left(C_t\right)\right)^- \dd t = +\infty\;\; \as{\P},
\end{equation}
and hence $J(C) = -\infty$. Here we only consider the case that $r < 0$; the case $r = 0$ (which implies that $\lambda = 0$) is similar but easier.

First, we show that for any $(\Pi, C) \in \cA(x)$, there exists $(\tilde \Pi, 0) \in \cA(x)$ such that
\begin{equation}
		\label{eq:log:kappa:neg:C bound}
e^{rt}\int_0^t e^{-ru}\ C_u \dd u \leq X^{x, \tilde \Pi, 0}_t \;\; \as{ \P}, \quad t \geq 0.
\end{equation}
Let $\theta_t := \frac{\Pi_t X^{x, \Pi, C}_t}{S_t}$ and define the process $ G = (G_t)_{t \geq 0}$ by $G_t = x_0 + \int_0^t {\theta_u S_u} (\sigma \dd W_u +  (\mu-r) \dd u)$.
Then using the product formula and the dynamics of $X^{x, \Pi, C}$, we obtain
\begin{align}
0 & \leq e^{-rt} X^{x, \Pi, C}_t = x_0 + \int_0^t \theta_u S_u \sigma (\dd W_u +  \lambda \dd u) - \int_0^t e^{-ru} C_u \dd u \\
&= G_t - \int_0^t e^{-ru}  C_u \dd u = e^{-rt} X^{x, \tilde \Pi, 0}_t - \int_0^t e^{-ru}  C_u \dd u,
\end{align}
where $\tilde \Pi_t :=  \tfrac{ \theta_t S_t}{G_t} \1_{ \{G_t>0\}} $.
%\footnote{Strictly speaking, $\tilde \Pi_t$ is not defined for $Y_t =0$, but this does not matter. We can for example set $\tilde \Pi_t := 0$ for $Y_t =0$.}
Rearranging gives \eqref{eq:log:kappa:neg:C bound}.

Next, we show that
\begin{equation}
\label{eq:log:LIL}
\liminf_{t \to \infty} X^{x, \tilde \Pi, 0} = 0\;\;\as{\P}
\end{equation}
It\^o's formula and the fact that $\kappa \leq 0$ gives
\begin{equation*}
\log(X^{x, \tilde \Pi, 0}) = \log(x_0) + \int_0^t \tilde \Pi_u  \sigma \dd W_u + \int_0^t \left(\kappa -\frac{(\tilde \Pi_u\sigma - \lambda)^2}{2}\right)\dd u \leq \log(x_0) + \int_0^t \tilde \Pi_u  \sigma \dd W_u.
\end{equation*}
There are two cases: On $\{\int_0^\infty \tilde \Pi_u^2 \dd u< \infty\}$, $\int_0^\infty (\kappa - \tfrac{(\tilde \Pi_u \sigma- \lambda)^2}{2}) \dd u = -\infty$ and  $\lim_{t \to \infty} \int_0^t \tilde \Pi_u  \sigma \dd W_u$ exists in $\RR$, whence $\lim_{t \to \infty} \log(X_t) = -\infty$. On $\{\int_0^\infty \tilde \Pi_u^2 \dd u = \infty\}$, ${\liminf_{t \to \infty} \int_0^t \tilde \Pi_u  \sigma \dd W_u = -\infty}$ by the law of iterated logarithm, whence $\liminf_{t \to \infty} \log(X_t) = -\infty$. So we have \eqref{eq:log:LIL}.

Finally, combining \eqref{eq:log:LIL} with \eqref{eq:log:kappa:neg:C bound} yields
\begin{equation}
\liminf_{t \to \infty}	e^{r t} \int_0^t e^{-r u} C_u \dd u = 0\;\; \as{\P}
\end{equation}
This implies that the random set $A := \{u \in [0, \infty) : C_u < 1/2\}$ has $\as{\P}$ infinite Lebesgue-measure. But this implies that
\begin{equation}
\int_0^\infty e^{-\delta t} \log\left(C_t\right)^- \dd t \geq \int_0^\infty \log\left(C_t\right)^-  \dd t\geq \int_A \log(2) \dd t = +\infty \;\; \as{\P}
\end{equation}

Putting the results together we have the following result for logarithmic utility:
\begin{thm}
\label{thm:log}
Suppose $\delta > 0$. %Let the function $\hat{V}: (0, \infty) \to \RR$ be given by $\hat V(x) = $.
Then for $x > 0$,
\begin{align}
    V^*(x) &:= \sup_{C \in \cC^*(x)}J(C) = J(\hat{C}) = \frac{1}{\delta^2}\left(\delta\log(\delta x) + r + \frac{\lambda^2}{2} - \delta\right),
\shortintertext{where the corresponding optimal investment-consumption strategy is given by $(\Pi,C)= (\hat{\Pi},\hat{C})$, where}
    \hat{\Pi} &=\frac{\lambda}{\sigma R}, \quad \hat{C} = \delta  X^{x,\hat \Pi, \hat C}.
    \label{eq:candlog}
\end{align}
Suppose $\delta \leq 0$. Then the problem is ill-posed. For $\kappa := r + \frac{\lambda^2}{2}>0$ we have $V^*(x) =+\infty$, whereas for $\kappa \leq 0$ we have $V^*(x)=-\infty$.
\end{thm}

\begin{comment}
Note that this case is not normally considered, and we are only considering it because we have done so for $R \neq 1$. In this case it is easily possible\footnote{For example, suppose $\kappa:= r + \frac{\lambda^2}{2}>0$ and consider the constant proportional strategy with $\pi = \hat{\pi}$ and $\xi = \kappa$. Then $\log(\xi X_t) = \log (\xi x) + \lambda W_t$ and $\int_0^\infty e^{-\delta t} (\log (\xi X_t))^+ dt = \infty = \int_0^\infty e^{-\delta t} (\log (\xi X_t))^- dt$.} to have $\int_0^\infty e^{-\delta t} (\log (\xi X_t))^+ dt = \infty = \int_0^\infty e^{-\delta t} (\log (\xi X_t))^- dt$. As a first step towards analysing the problem in the case it is necessary to decide what value to assign to such a consumption/investment strategy. This is not the focus of this paper, and so we will not pursue the discussion here.
\end{comment}

\section{Change of numéraire arguments and the role of $\delta$}
\label{sec:numeraire}
{It is interesting to study how the Merton problem behaves under a {\emph {change of numéraire}}. As we have seen in Appendix~\ref{sec:literatureR>1}, using the perturbation arguments of Karatzas et al~\cite{karatzas1986explicit} or Davis and Norman \cite{davis1990portfolio}, we get verification arguments for the case $R > 1$ under the parameter restrictions $\delta > 0$ and $r > 0$. The goal of this section is to show using a {change of numéraire} that this parameter restriction can be weakened, although not to the extent that it covers all the parameter combinations for which $\eta>0$. We then discuss how these arguments shed some light on the interpretation of the parameter $\delta$.}

A pair $(\tilde S^0, \tilde S) = (\tilde S^0_t, \tilde S_t)_{t \geq 0}$ of semimartingales is said to be \emph{economically equivalent} to $(S^0, S)$ if there exists a positive continuous semimartingale $D = (D_t)_{t\geq 0}$ such that $\tilde S^0 = D S^0$ and $\tilde S = D S$. Here, the interpretation of $D$ is an exchange rate process and $(\tilde S^0, \tilde S)$ describes the financial market in a different currency unit; see \cite[Section 2.1]{herdegen:17} for more details. We will restrict attention to deterministic processes $D$ in which case $D$ is better described as a change in accounting units.

Next, recall that if $(\theta^0, \theta, C)$ is a admissible investment-consumption strategy for initial wealth $x > 0$, (where $\theta^0$ and $\theta^1$ denote the \emph{number of shares} held in the riskless and risky asset, respectively), then the corresponding wealth process $X = \theta^0_t S^0_t+ \theta_t S$ satisfies the SDE
\begin{equation}
	\dd X_t = \theta^0_t \dd S^0_t+ \theta_t \dd S_t - C_t \dd t.
\end{equation}
Now if $(\tilde S^0, \tilde S)$ is economically equivalent to $(S^0, S)$ with corresponding exchange rate process $D$, it is not difficult to check that the corresponding wealth process $\tilde X := \theta^0 \tilde S^0 + \theta \tilde S = D X$ satisfies the SDE
\begin{align}
	\diff \tilde X_t
	% & = D_t \dd X_t + X_t \dd D_t + \langle X, D \rangle_t \\
	%&= \theta^0_t D_t \dd S^0_t + \theta_t D_t \dd S_t -  D_t C_t \dd t + (\theta^0_t S^0_t + \theta_t S_t) \dd D_t  + \theta^0_t \dd \langle S^0, D \rangle_t + \theta_t \dd \langle S, D \rangle_t -d \langle C, D\rangle_t\\
	&= \theta^0_t \dd \tilde S^0 + \theta_t \dd \tilde S- \tilde C_t \dd t,
\end{align}
where $\tilde C = D C$. This means that if $C$ describes an attainable consumption strategy in units corresponding to $(S^0, S)$, then $\tilde C = D C$ describes the \emph{same} consumption strategy in units corresponding to $(\tilde S^0, \tilde S)$ (which is also attainable for those units).

Consider now the case that $D_t = e^{\gamma t}$ for some $\gamma \in \RR$. Then $(\tilde S^0, \tilde S)$ is again a Black-Scholes-Merton model with interest rate $\tilde r = r + \gamma$, drift $\tilde \mu = \mu + \gamma$ and volatility $\tilde \sigma = \sigma$. Let $C$ be an attainable consumption strategy in units corresponding to $(S^0, S)$ and $\tilde C = D C$ the corresponding attainable consumption strategy in units corresponding to $(\tilde S^0, \tilde S)$. Then
$\tilde{C}/\tilde{S}^0 = DC/D S^0 = C/S^0$ and
\begin{align}
	J(C; \delta) ~~:=&~~~ \EX{\int_0^\infty \frac{e^{-\delta t}}{1-R} C^{1-R}_t \dd t} = \EX{\int_0^\infty \frac{e^{-(\delta + r(R-1))t}}{1-R}  \left(\frac{C_t}{S^0_t}\right)^{1-R}\dd t} \label{eq:Jnumeraire1} \\
	=&~~~ J(C/S^0; \delta + r(R-1)) =J(\tilde C/\tilde S^0; \delta + (\tilde{r}-\gamma)(R-1)) \nonumber \\
	=&~~~ \EX{\int_0^\infty \frac{e^{-(\delta - (R-1) \gamma + \tilde r (R-1))t}}{1-R}  \left(\frac{\tilde C_t}{\tilde S^0_t}\right)^{1-R}\dd t} = \EX{\int_0^\infty \frac{e^{-(\delta - (R-1) \gamma) t}}{1-R} \tilde C^{1-R}_t \dd t} \nonumber \\
	=&~~~ J(\tilde C; \delta - (R-1) \gamma). \label{eq:Jnumeraire2}
\end{align}

It follows from the above calculation that the Merton problem for $R, r, \mu, \sigma, \delta$ is equivalent to the Merton problem for $R, r + \gamma, \mu + \gamma,$ $\sigma, \delta - (R-1) \gamma$ for each $\gamma \in \RR$. Note in particular, that the well-posedness parameter $\eta$ from \eqref{eqn: definition of eta} is independent of the choice of accounting units.
This means that if we have a verification argument for the parameters $R, r + \gamma, \mu + \gamma, \sigma, \delta - (R-1) \gamma$, we also have verification argument for the parameters $R, r, \mu, \sigma, \delta$. Hence, if $\delta + r(R-1)>0$ we can choose $\gamma = \frac{\delta - r(R-1)}{2(R-1)}$ so that $\tilde{\delta} = \tilde{r} = \frac{\delta + r(R-1)}{2}>0$ and then we can
extend the verification arguments of Karatzas et al~\cite{karatzas1986explicit} or Davis and Norman \cite{davis1990portfolio} to this case. It follows that instead of needing to assume $\delta>0$ and $r>0$ as in \cite{karatzas1986explicit} and \cite{davis1990portfolio} it is sufficient to assume only that $\delta + r(R-1) > 0$.

Nonetheless, the condition $\delta + r(R-1) > 0$ is stronger than the condition for a well-posed problem (namely $\eta>0$) and there are parameter values which we would like to consider (and which are covered by Theorem~\ref{thm:veri}) for which the verification arguments of \cite{karatzas1986explicit} and \cite{davis1990portfolio} do not apply, even after the change of num\'{e}raire arguments of this section.

\medskip{}
The above ideas also shed some light on the interpretation of the parameter $\delta$. To this end, consider an alternative formulation of the Merton problem and associate to an attainable consumption stream $C$ the expected utility
	\begin{equation}
		K(C; \phi) = \EX{\int_0^\infty \frac{e^{-\phi t}}{1-R}  \left(\frac{C_t}{S^0_t}\right)^{1-R} \dd t},
		\label{eq:Kdef}
	\end{equation}
	where $\phi:=\delta + r(R-1)$ is the {\em impatience rate}. Then $K(C;\phi) = J(C, \phi - r(R-1))$. In order to emphasise the dependence of the problem on the accounting units which are being used we might expand the notation to write $J(C;S^0,S;\delta)$ and $K(C;S^0,S;\phi)$ and then \eqref{eq:Jnumeraire2} becomes
	\[ J(C;S^0,S;\delta) = J(\tilde{C};\tilde{S}^0,\tilde{S};\delta-(R-1)\gamma), \]
	whilst, for $K(C,\phi) = K(C; S^0,S,\phi)$ we find
	\[ K(\tilde{C}; \tilde{S}^0,\tilde{S},\phi) = \EX{\int_0^\infty \frac{e^{-\phi t}}{1-R}  \left(\frac{D_t C_t}{D_t S^0_t}\right)^{1-R} \dd t}
	= K({C};{S}^0,{S},\phi). \]
	In particular, $K$ defined via \eqref{eq:Kdef} has the advantage that (unlike $J$) it is num\'{e}raire-independent in the sense that a change of accounting unit leaves the problem value unchanged.\footnote{This applies not only to deterministic changes of accounting units, but also to stochastic changes of numéraire.} With this in mind it makes sense to focus on the impatience rate $\phi$ rather than the discount rate $\delta$.
	Note that $\eta = \frac{1}{R}\phi + \frac{(R-1)}{R} \frac{\lambda^2}{2R}$ so that the optimal consumption rate is a linear (convex if $R>1$) combination of the impatience rate and (half of) the squared Sharpe ratio per unit of risk aversion, with the weights depending on the risk aversion.

\section{The dual approach}

For completeness, we include a brief description of the dual approach to the Merton problem. This is a static argument in the sense that we replace the dynamic admissibility condition --- that the wealth process is negative at all times --- with a static \textit{budget feasibility} condition. As before, we only deal with a Black\textendash Scholes\textendash Merton financial market, and in this case the argument is particular simple since the market is complete and hence there is exactly one equivalent martingale measure.

Define the state-price density process $\zeta = (\zeta_t)_{t\geq0}$ by
\begin{equation}
\label{def:spd}
\zeta_t = e^{- r t} \cE(\lambda W)_t = \exp \left( \lambda W_t - \left( r + \frac{\lambda^2}{2} \right) t \right).
\end{equation}
The following proposition gives a neat equivalent criterion for a consumption process to be admissible in terms of the state price density.

\begin{proposition}
\label{prop:admissiblityandbudgetfeasibilityequivalent}
	A non-negative progressively measurable process $C$ is in $\cC(x)$ if and only if the \textit{budget feasibility condition} holds:
	\begin{equation}\label{eq:budget feasibility condition}
		\EX{\int_0^\infty \zeta_s C_s \dd s}\leq x.
	\end{equation}
\end{proposition}

\begin{proof}
We only prove necessity of \eqref{eq:budget feasibility condition} since that is all that we require for our subsequent arguments. Sufficiency can be proved using the the Brownian martingale representation theorem. % (see, for example, \cite[IV.3 Corollary 2]{protter2005stochastic}). %See \cite{???} for more details.

	Suppose that $C \in \cC(x)$. Let $\Pi$ be the corresponding investment process such that $(\Pi,C) \in \cA(x)$. Denote the corresponding wealth process by $X = X^{x,\Pi, C}$ and define the
nonnegative process $Y = (Y_t)_{t \geq 0}$ by
	$Y_t = \xi_t X_t + \int_0^t \zeta_s C_s \dd s.$
The product rule, the definition of  $\zeta$ in \eqref{def:spd} and the dynamics of $X$ from \eqref{eqn:original wealth process} give
	\begin{equation}
\label{eq:dY}
	\dd Y_t = \zeta_t \dd X_t + X_t \dd \zeta_t + \dd \langle \zeta, X \rangle_t + \zeta_t C_t \dd t
		= \zeta_t X_t (\Pi_t \sigma - \lambda) \dd W_t.
	\end{equation}
Hence, $Y$ is a nonnegative local martingale and therefore a supermartingale. Using $\zeta_t X_t \geq 0$ and Fatou's lemma we conclude that
	\begin{equation}
		x = Y_0 \geq \liminf_{t\to\infty}\EX{Y_t} \geq \EX{\int_0^\infty \zeta_s C_s \dd s}.
	\end{equation}
This ends the proof.
\begin{comment}	
	Conversely, suppose that Equation \eqref{eq:budget feasibility condition} holds. Then, since $\xi = \int_0^\infty \zeta_t C_t \dd t$ is an integrable random variable, we may define the uniformly integrable martingale $Z = (Z_t)_{t\geq0}$ by $Z_t = \E[\xi|\cF_t]$. By the Brownian martingale representation theorem (see, for example, \cite[IV.3 Corollary 2]{protter2005stochastic}), there exists a predictable process $H = (H_t)_{t\geq0}$ such that
	\begin{equation}
		Z_t = \EX{\xi} + \int_0^t H_s \dd W_s.
	\end{equation}
	Let $\Pi =(\Pi_t)_{t\geq0}$ be given by
	\begin{equation}
		\Pi_t X_t = \frac{H_t + \lambda \zeta_t}{\zeta_t\sigma}.
	\end{equation}
	Then, using \eqref{eq:Y state price density equation} we find that the wealth process started from initial wealth $x' = \EX{\xi}\leq x$ satisfies
	\begin{equation}
		\zeta_t X_t + \int_0^t \zeta_s C_s \dd s = Y_t = Z_t = \cEX[t]{\int_0^\infty \zeta_s C_s \dd s}.
	\end{equation}
	Hence, $\zeta_t X_t = \E[\int_t^\infty\zeta_s C_s \dd s|\cF_t] \geq0$ since both $\zeta$ and $C$ are non-negative. Consequently, the pair $(\Pi,C)$ is admissible and $C$ is an attainable consumption stream.
\end{comment}
\end{proof}
\begin{comment}
Using this equivalence, the Merton problem \eqref{eqn: value function definition power law case} becomes to find
\begin{equation}
	\sup_{C\geq0}\E \left[ \int_0^\infty e^{-\delta t} U(C_t) \dd t \right], \qquad \text{subject to}\qquad \EX{\int_0^\infty \zeta_t C_t \dd t}\leq x,
\end{equation}
for $U(c) = \frac{c^{1-R}}{1-R}$. %We have already found a candidate for the optimal strategy, but we derive it from the dual problem since it highlights the essence of the approach. First,
Let $y>0$ be a Lagrange multiplier. Then, the Lagrangian is given by
\begin{equation}
 \E\left[ \int_0^\infty \left(e^{-\delta t} U(C_t) dt - y \left(\int_0^\infty \zeta_t C_t dt - x \right) \right) \right] =	\E \left[\int_0^\infty \left(e^{-\delta t} U(C_t) - y \zeta_t C_t \right)\dd t \right] + x y.
\end{equation}
We then optimise inside the integral taking derivatives to find that the candidate optimal strategy is given by $\hat{C}^y_t = I(e^{\delta t} y \zeta_t)$, where $I(y) = (U')^{-1}(y)$. It is now possible to choose $\hat{y}$ to make the constraint hold for $\hat{C}^{\hat{y}}$ with equality: $\E[\int_0^\infty \zeta_t \hat{C}^{\hat{y}}_t \dd t] =x$. It is an easy exercise to show that $\hat{C} = \hat{C}^{\hat{y}}$ is the same as the consumption process defined in \eqref{eq:cand}. Moreover it follows that for $\hat{C}$ we have $\E \left[ \int_0^\infty e^{- \delta t} \frac{\hat{C}_t^{1-R}}{1-R} \right] = \hat{V}(x)$.
\end{comment}

The verification argument then goes as follows. First, the candidate consumption process $\hat{C}$ from \eqref{eq:cand} satisfies
\begin{equation}
		\label{eq:spd}
	e^{-\delta t} U'(\hat{C}_t) = {\hat V}_x(x)\zeta_t,
\end{equation}
as well as
\begin{equation}
	\label{eq:budget:Chat}
	\EX{\int_0^\infty \zeta_s \hat C_s \dd s} = x.
\end{equation}
Then for an admissible $C\in \cC(x)$, using the budget condition \eqref{eq:budget feasibility condition} for $C$ and as well as the budget condition 	\eqref{eq:budget:Chat} for $\hat C$, together with \eqref{eq:spd} and the simple fact that the concave function $U$ is bounded above by its tangent, i.e., $U(a) \leq U(a) + (b-a)U'(a)$ for $a >0, b\geq0$, we obtain
\begin{align}
	\EX{\int_0^\infty e^{-\delta t} U(C_t) \dd t} \leq&~\EX{\int_0^\infty \left( e^{-\delta t} U(\hat{C}_t) + e^{-\delta t} U'(\hat{C}_t)(C_t - \hat{C}_t) \right)\dd t}
	\\ \label{eq:duality inequality}
	=&~ \EX{\int_0^\infty e^{-\delta t} U(\hat{C}_t) \dd t} + {\hat V}_x (x) \EX{\int_0^\infty \zeta_t(C_t - \hat{C}_t) \dd t} \\
 \leq &~ \EX{\int_0^\infty e^{-\delta t} U(\hat{C}_t) \dd t} = \hat{V}(x).
\end{align}

%Here, the last inequality from the fact that since since $C$ is admissible, it is budget feasible, and that the budget constraint in \eqref{eq:budget feasibility condition} is satisfied with equality for $\hat{C}$.
This proves optimality of $\hat{C}$.
%If one wishes, they may find the corresponding investment plan $\hat{\Pi}$ by the approach used in the proof of Proposition \ref{prop:admissiblityandbudgetfeasibilityequivalent}.

\end{document}